\documentclass[runningheads]{llncs}
\usepackage{cases} 
\usepackage{complexity}
\usepackage{amsfonts,amstext}
\usepackage{algorithmic}
\usepackage[ruled, algo2e, vlined, english]{algorithm2e}
\usepackage{graphicx}
\newtheorem{nclaim}[theorem]{Claim}

\newenvironment{claimproof}{\paragraph{\textit{Proof.}}}{\hfill$\diamondsuit$}

\newcommand{\opt}{\mathrm{opt}}

\newcommand{\cdn}{\mathrm{cdn}}

\newcommand{\mod}{\mathrm{~mod~}}

\begin{document}
\title{Linear-time Algorithms for Eliminating Claws in Graphs\thanks{The authors would like to thank CAPES, FAPERJ, CNPq, ANPCyT, and UBACyT for the partial support.}}
%
%
\author{Flavia Bonomo-Braberman\inst{1} \and
Julliano R. Nascimento\inst{2} \and Fabiano S. Oliveira\inst{3}
\and U\'{e}verton S. Souza\inst{4} \and Jayme L.
Szwarcfiter\inst{3,5}}
\authorrunning{Bonomo, Nascimento, Oliveira, Souza, and Szwarcfiter}
%
\institute{Universidad de Buenos Aires. FCEyN. DC. / CONICET-UBA.
ICC. Argentina. \email{fbonomo@dc.uba.ar} \and INF, Universidade
Federal de Goi\'{a}s, GO, Brazil. \email{julliano@inf.ufg.br} \and
IME, Universidade do Estado do Rio de Janeiro, RJ, Brazil.
\email{fabiano.oliveira@ime.uerj.br} \and IC, Universidade Federal
Fluminense, RJ, Brazil. \email{ueverton@ic.uff.br} \and IM, COPPE,
and NCE, Universidade Federal do Rio de Janeiro, RJ, Brazil.
\email{jayme@nce.ufrj.br}}
\maketitle              
\begin{abstract}
Since many $\NP$-complete graph problems have been
shown polynomial-time solvable when restricted to claw-free
graphs, we study the problem of determining the
distance of a given graph to a claw-free graph, considering vertex
elimination as measure.
\textsc{Claw-free Vertex Deletion (CFVD)} consists of determining
the minimum number of vertices to be removed from a graph such
that the resulting graph is claw-free.
Although \textsc{CFVD} is $\NP$-complete in general and 
recognizing claw-free graphs is still a challenge, where the
current best algorithm for a graph $G$ has the same running time of the
best algorithm for matrix multiplication, we present linear-time algorithms for
\textsc{CFVD} on weighted block graphs and weighted graphs with
bounded treewidth. Furthermore, we show that this problem can be
solved in linear time by a simpler algorithm on forests, and we
determine the exact values for full $k$-ary trees.
On the other hand, we show that \textsc{Claw-free Vertex Deletion}
is $\NP$-complete even when the input graph is a split graph. We
also show that the problem is hard to approximate within any
constant factor better than $2$, assuming the Unique Games
Conjecture.

\end{abstract}
\section{Introduction}
\label{sec:introduction} In 1968, Beineke~\cite{Beineke68}
introduced claw-free graphs as a generalization of line graphs.
Besides that generalization, the interest in studying the class of
claw-free graphs also emerged due to the results showing that some
$\NP$-complete problems are polynomial time solvable in that class
of graphs. For example, the maximum independent set problem is
polynomially solvable for claw-free graphs, even on its weighted
version~\cite{F-O-S-claw-free-acm}.

A considerable amount of literature has been published on
claw-free graphs. For instance, Chudnovsky and Seymour provide a
series of seven papers describing a general structure theorem for
that class of graphs, which are sketched in~\cite{C-S-clawfree}.
Some results on domination, Hamiltonian properties, and matchings
are found
in~\cite{hedetniemi1988recent},~\cite{flandrin1993hamiltonian},~and~\cite{sumner1974graphs},
respectively. In the context of parameterized complexity, Cygan et
al.~\cite{cygan2011dominating} show that finding a minimum
dominating set in a claw-free graph is fixed-parameter tractable.
For more on claw-free graphs, we refer to a survey by Faudree,
Flandrin and Ryj\'a\v{c}ek~\cite{faudree1997claw} and references
therein.


The aim of our work is to obtain a claw-free graph by a minimum
number of vertex deletions. Given a graph $G$ and a property
$\Pi$, Lewis and Yannakakis~\cite{L-Y-node-deletion} define a
family of vertex deletion problems (\textsc{$\Pi$-Vertex
Deletion}) whose goal is finding the minimum number of vertices
which must be deleted from $G$ so that the resulting graph
satisfies $\Pi$. Throughout this paper we consider the property
$\Pi$ as belonging to the class of claw-free graphs. For a set $S
\subseteq V(G)$, we say that $S$ is a \textit{claw-deletion set}
of $G$ if $G \setminus S$ is a claw-free graph.



We say that a class of graphs $\mathcal{C}$ is \emph{hereditary}
if, for every graph $G \in \mathcal{C}$, every induced subgraph of
$G$ belongs to $\mathcal{C}$. If either the number of graphs in
$\mathcal{C}$ or the number of graphs not in $\mathcal{C}$ is
finite, then $\mathcal{C}$ is \emph{trivial}. A celebrated result
of Lewis and Yannakakis~\cite{L-Y-node-deletion} shows that for
any hereditary and nontrivial graph class $\mathcal{C}$,
\textsc{$\Pi$-Vertex Deletion} is $\NP$-hard for $\Pi$ being the
property of belonging to $\mathcal{C}$. Therefore,
\textsc{$\Pi$-Vertex Deletion} is $\NP$-hard when $\Pi$ is the
property of belonging to the class $\mathcal{C}$ of claw-free
graphs.
Cao et al.~\cite{cao2018vertex} obtain several results when $\Pi$
is the property of belonging to some particular subclasses of
chordal graphs. They show that transforming a split graph into a
unit interval graph with the minimum number of vertex deletions
can be solved in polynomial time. In contrast, they show that
deciding whether a split graph can be transformed into an interval
graph with at most $k$ vertex deletions is $\NP$-complete.
Motivated by the works of Lewis and
Yannakakis~\cite{L-Y-node-deletion} and Cao et
al.~\cite{cao2018vertex}, since claw-free graphs is a natural
superclass of unit interval graphs, we study vertex deletion
problems associated with eliminating claws. The problems are
formally stated below.

\begin{problem}{\textsc{Claw-free Vertex Deletion (CFVD)}}\\
Instance: A graph $G$, and $k \in \mathbb{Z}^{+}$.\\
Question: Does there exist a claw-deletion set $S$ of $G$ with
$|S| \leq k$?
\end{problem}

\begin{problem}{\textsc{Weighted Claw-free Vertex Deletion (WCFVD)}}\\
Instance: A graph $G$, a weight function $w: V(G) \to \mathbb{Z}^{+}$, and $k \in \mathbb{Z}^{+}$.\\
Question: Does there exist a claw-deletion set $S$ of $G$ with
$\sum_{v \in S} w(v) \leq k$?
\end{problem}

By Roberts' characterization of unit interval
graphs~\cite{Rob-uig}, \textsc{Claw-free Vertex Deletion} on
interval graphs is equivalent to the vertex deletion problem where
the input is restricted to the class of interval graphs and the
target class is the class of unit interval graphs, a long standing
open problem (see e.g.~\cite{cao2018vertex}). Then, the results by
Cao et al.~\cite{cao2018vertex} imply that \textsc{Claw-free
Vertex Deletion} is polynomial-time solvable when the input graph
is in the class of interval $\cap$ split graphs. Moreover, their
algorithm could be also generalized to the weighted version. In
this paper, we show that \textsc{Claw-free Vertex Deletion} is
$\NP$-complete when the input graph is in the class of split
graphs.

The results by Lund and Yannakakis~\cite{L-Y-approx} imply that
\textsc{Claw-free Vertex Deletion} is \emph{APX}-hard and admits a
$4$-approximating greedy algorithm. Even for the weighted case, a
pricing primal-dual $4$-approximating algorithm is known for the
more general problem of $4$-\textsc{Hitting
Set}~\cite{Hoch-approx}. The \textsc{CFVD} problem is
$\NP$-complete on bipartite graphs~\cite{Yan-node-deletion-bip},
and a $3$-approximating algorithm is presented by Kumar et al.
in~\cite{Kumar-bip-clawdel} for weighted bipartite graphs. We
prove that the unweighted problem is hard to approximate within
any constant factor better than~$2$, assuming the Unique Games
Conjecture, even for split graphs.

Regarding to parameterized complexity, \textsc{Claw-free Vertex
Deletion} is a particular case of \textsc{$H$-free Vertex
Deletion}, which can be solved in $|V(H)|^k n^{\mathcal{O}(1)}$
time using the bounded search tree technique. In addition, it can
also be observed that \textsc{CFVD} is a particular case of
$4$-\textsc{Hitting Set} thus, by Sunflower lemma, it admits a
kernel of size $\mathcal{O}(k^4)$, and the complexity can be
slightly improved~\cite{Fernau10-hitsetw}. With respect to width
parameterizations, it is well-known that every optimization
problem expressible in LinEMSOL$_1$ can be solved in linear time
on graphs with bounded cliquewidth~\cite{C-M-cw-LinEMSOL}. Since
claws are induced subgraphs with constant size, it is easy to see
that finding the minimum weighted $S$ such that $G\setminus S$ is
claw-free is LinEMSOL$_1$-expressible. Therefore, \textsc{WCFVD}
can be solved in linear time on graphs with bounded cliquewidth,
which includes trees, block graphs and bounded treewidth graphs.
However, the linear-time algorithms based on the MSOL model-checking framework~\cite{courcelle1990monadic} typically do not provide useful algorithms in practice since the dependence on the cliquewidth involves huge
multiplicative constants, even when the clique-width is bounded by two (see~\cite{flum2006parameterized}). 
In this work, we provide explicit
discrete algorithms to effectively solve \textsc{WCFVD} in linear
time in practice on block graphs and bounded treewidth graphs.
Even though forests are particular cases of bounded treewidth
graphs and block graphs, we describe a specialized simpler
linear-time algorithm for \textsc{CFVD} on forests. This allows us
to determine the exact values of \textsc{CFVD} for a full $k$-ary
tree $T$ with $n$ vertices. If $k = 2$, we show that a minimum
claw-deletion set of $T$ has cardinality $(n+1-2^{(\log_2 (n+1)
\!\! \mod 3)})/{7}$, and $(nk-n+1-k^{(\log_k (nk-n+1) \!\! \mod
2)})/({k^2-1})$, otherwise.

This paper is organized as follows. Section~\ref{sec:complexity}
is dedicated to show the hardness and inapproximability results.
Sections~\ref{sec:forests}, \ref{sec:blockGraphs},
and~\ref{sec:treewidth} present results on forests, block graphs,
and bounded treewidth graphs, res\-pecti\-ve\-ly. Due to space
constraints, proofs of statements marked with `$\clubsuit$' are
deferred to the appendix, as well as some additional results and
well known definitions.


\paragraph*{\bf Preliminaries.} We consider simple and undirected graphs, and
we use standard terminology and notation.

Let $T$ be a tree rooted at $r \in V(T)$ and $v \in V(T)$. We denote by $T_v$ the subtree of $T$ rooted at $v$, and by $C_T(v)$ the set of children of $v$ in $T$.
For $v \neq r$, denote by $p_T(v)$ the parent of $v$ in $T$, and by $T_v^+$ the subgraph of $T$ induced by $V(T_v) \cup \{p_T(v)\}$. Let $T_r^+ = T$ and $p_T(r) = \emptyset$.
When $T$ is clear from the context, we simply write $p(v)$ and $C(v)$.

The \textit{block-cutpoint-graph} of a graph $G$ is the bipartite
graph whose vertex set consists of the set of cutpoints of $G$ and
the set of blocks of $G$. A cutpoint is adjacent to a block
whenever the cutpoint belongs to the block in $G$. The
block-cutpoint-graph of a connected graph is a tree and can be
computed in $\mathcal{O}(|V(G)|+|E(G)|)$ time~\cite{Tarjan-DFS}. 

Let $G$ and $H$ be two graphs. We say that $G$ is
\textit{$H$-free} if $G$ does not contain a graph isomorphic to
$H$ as an induced subgraph. A \textit{claw} is the complete
bipartite graph $K_{1,3}$. The class of \emph{linear forests} is
equivalent to that of claw-free forests. A vertex $v$ in a claw
$C$ is a \textit{center} if $d_C(v) = 3$. The cardinality
$\cdn(G)$ of a minimum claw-deletion set in $G$ is the
\textit{claw-deletion number} of $G$. For our proofs, it is enough
to consider connected graphs, since a minimum (weight)
claw-deletion set of a graph is the union of minimum (weight)
claw-deletion sets of its connected components.
Williams et al.~\cite{williams2015finding} show that induced claws 
in an $n$-vertex graph $G$ can be detected in $\mathcal{O}(n^{\omega})$ time, 
where $\omega$ is the matrix multiplication exponent. As far as we know, the best upper bound is $\omega <
2.3728639$~\cite{legall2014powers}.

\section{Complexity and Approximability Results}
\label{sec:complexity}

The result of Lewis and Yannakakis~\cite{L-Y-node-deletion}
implies that \textsc{Claw-free Vertex Deletion} is $\NP$-complete.
In this section, we show that the same problem is $\NP$-complete
even when restricted to split graphs, a well known subclass of
chordal graphs.
Before the proof, let us recall that the \textsc{Vertex Cover
(VC)} problem consists of, given a graph $G$ and a positive
integer $k$ as input, deciding whether there exists $X \subseteq
V(G)$, with $|X| \leq k$, such that every edge of $G$ is incident
to a vertex in $X$.

\begin{theorem}\label{theo:clawDeletionSplit}
\textsc{Claw-free Vertex Deletion} on split graphs is $\NP$-complete.
\end{theorem}

\begin{proof}
\textsc{Claw-free Vertex Deletion} is clearly in $\NP$ since claw-free graphs can be recognized in polynomial time~\cite{williams2015finding}.
To show $\NP$-hardness, we employ a reduction from $\textsc{Vertex Cover}$ on general graphs~\cite{G-J}. 

Let $(G,k)$ be an instance of vertex cover, where $V(G) = \{v_1,
\dots, v_n\}$, and $E(G) = \{e_1, \dots, e_m\}$. Construct a split
graph $G' = (C \cup I, E')$ as follows. The independent set is $I
= \{v'_1, \dots, v'_n\}$. The clique $C$ is partitioned into sets
$C_i$, $1 \leq i \leq m+1$, each on $2n$ vertices. Given an
enumeration $e_1, \dots, e_m$ of $E(G)$, if $e_{i} = v_jv_\ell$,
make $v'_j$ and $v'_{\ell}$ adjacent to every vertex in $C_i$.


We prove that $G$ has a vertex cover of size at most $k$ if and only if $G'$ has a claw-deletion set of size at most $k$. We present Claim~\ref{claimTwoVertices} first.

\begin{nclaim} \label{claimTwoVertices}
Every claw in $G'$ contains exactly two vertices from $I$.
\end{nclaim}

\begin{claimproof}
Let $C'$ be a claw in $G'$. Since $C' \cap C$ is a clique, $|C'
\cap C| \leq 2$, thus $|C' \cap I| \geq 2$ and the center of the
claw must be in $C$. On the other hand, by construction, $d_{I}(u)
= 2$ for every $u \in \bigcup_{i = 1}^{m} C_i$. This implies $|C'
\cap I| \leq 2$.
\end{claimproof}

Suppose that $X$ is a vertex cover of size at most $k$ in $G$.
Then, every edge of $G$ is incident to a vertex in $X$. Let $e_i
\in E(G)$ and $X' = \{v' : v \in X\}$. By construction, every
vertex in $C_i$ is adjacent to a vertex in $X'$, therefore $|N_{G'
\setminus X'}(C_i) \cap I| \leq 1$. It follows by
Claim~\ref{claimTwoVertices} that $G' \setminus X'$ is claw-free.

Now, suppose that $S'$ is a claw-deletion set of $G'$ of size at
most $k$. Recall that $|C_i| = 2n$, for every $1 \leq i \leq m+1$.
Since $|S'| \leq k$, it follows that there exist $w_i \in C_i
\setminus S'$, for every $1 \leq i \leq m+1$. Let $1 \leq i \leq
m$ and $N_{I}(w_i) = \{u',v'\}$. Note that $\{u',v',w_i,
w_{m+1}\}$ induces a claw in $G'$. Since $S'$ is a claw-deletion
set of $G'$, we have that $S' \cap \{u',v'\} \neq \emptyset$. Let
$S = \{v: v' \in S' \cap I\}$. By construction, every $uv \in
E(G)$ is incident to a vertex in $S$, thus $S$ is a vertex cover
of $G$. \qed\end{proof}


Theorem \ref{theo:4aprox} provides a lower bound for the
approximation factor of \textsc{CFVD}. For terminology not defined
here, we refer to Crescenzi~\cite{crescenzi1997short}.

\begin{theorem}\label{theo:4aprox}
\textsc{Claw-free Vertex Deletion} cannot be approximated with
$2-\varepsilon$ ratio for any $\varepsilon > 0$, even on split
graphs, unless \emph{Unique Games Conjecture} fails.
\end{theorem}

\begin{proof}
The \emph{Unique Games Conjecture} was introduced by Khot
\cite{khot2002power} in 2002. Some hardness results have been
proved assuming that conjecture, for instance, see
\cite{khot2015unique}.
%
Given that \textsc{Vertex Cover} is hard to approximate to within $2 - \varepsilon$ ratio for any $\varepsilon > 0$ assuming the Unique Games Conjecture~\cite{khot2002power}, we perform an approximation-preserving reduction from \textsc{Vertex Cover}. Let $G$ be an instance of \textsc{Vertex Cover}. Let $f(G) = G'$ where $G'$ is the instance of \textsc{Claw-free Vertex Deletion} constructed from $G$ according to the reduction of Theorem~\ref{theo:clawDeletionSplit}. From Theorem~\ref{theo:clawDeletionSplit} we know that $G$ has a vertex cover of size at most $k$ if and only if $G'$ has a claw-deletion set of size at most $k$. Recall that $k\leq n=|V(G)|$. Then, for every instance $G$ of \textsc{Vertex Cover} it holds that $\opt_{\textsc{CFVD}}(G') = \opt_{\textsc{VC}}(G)$.
%
%
Now, suppose that $S'$ is a $(2-\varepsilon)$-approximate solution of $G'$ for \textsc{CFVD}.
Recall that $|C_i| = 2n$, for every $1 \leq i \leq m+1$. Since $\opt_{\textsc{CFVD}}(G') = \opt_{\textsc{VC}}(G) \leq n$, it follows that $|S'| < 2n$, thus, there exists $x \in C_{m+1} \setminus S'$, and $w \in C_i \setminus S'$, for every $1 \leq i \leq m$.
Again, let $N_{I}(w) = \{u',v'\}$. Note that $\{u',v',w, x\}$
induces a claw in $G'$. Since $S'$ is a claw-deletion set of $G'$,
we have that $S' \cap \{u',v'\} \neq \emptyset$. Let $S = \{v: v'
\in S' \cap I\}$. By construction, every $uv \in E(G)$ is incident
to a vertex in $S$, and therefore $S$ is a vertex cover of $G$.
Since $|S|\leq |S'|$ and $S'$ is a $(2-\varepsilon)$-approximate
solution of $G'$, then $|S|\leq |S'|\leq (2-\varepsilon)\cdot
\opt_{\textsc{CFVD}}(G') =
(2-\varepsilon)\cdot\opt_{\textsc{VC}}(G)$. Therefore, if
\textsc{CFVD} admits a $(2-\varepsilon)$-approximate algorithm
then \textsc{Vertex Cover}  also admits a
$(2-\varepsilon)$-approximate algorithm, which implies that the
Unique Games Conjecture fails~\cite{khot2002power}.
\qed\end{proof}

\section{Forests}
\label{sec:forests}

We propose Algorithm~\ref{alg:forestClawDeletionSet} to compute a minimum claw-deletion set $S$ of a rooted tree $T$.
The correctness of such algorithm follows in Theorem~\ref{theo:clawLinearForest}.

\begin{center}
\begin{minipage}{\textwidth}
\begin{algorithm2e}[H]
\label{alg:forestClawDeletionSet}
 \algsetup{linenosize=\small}
 \small
 \DontPrintSemicolon
 \LinesNumbered
 \BlankLine
 \KwIn{A rooted tree $T$, a vertex $v$ of $T$, and the parent $p$ of $v$ in $T$.}
 \KwOut{A minimum claw-deletion set $S$ of $T_v^+$, such that: if $\cdn(T_v^+)=1+\cdn(T_v)$ then $p \in S$; if $\cdn(T_v^+)=\cdn(T_v)$ and $\cdn(T_v)=1+\cdn(T_v\setminus \{v\})$ then $v \in S$.}
 \BlankLine
 \eIf{$C(v) = \emptyset$}{
    \Return $\emptyset$\;
 }{
    $S := \emptyset$\;
    \ForEach{$u \in C(v)$}{
        $S := S \cup \textsc{Claw-Deletion-Set}(T, u, v)$\;
    }
    $c := |C(v) \setminus S|$\;
    \uIf{$c \geq 3$}{
         $S := S \cup \{v\}$\;
    }
    \ElseIf{$c = 2$ and $p \neq \emptyset$ and $v \notin S$}{
            $S := S \cup \{p\}$\;
    }
    \Return $S$\;
 }
\caption{\textsc{Claw-Deletion-Set}($T$, $v$, $p$)}
\end{algorithm2e}
\end{minipage}
\end{center}

\begin{theorem}\label{theo:clawLinearForestJ}$(\clubsuit)$
Algorithm~\ref{alg:forestClawDeletionSet} is correct. Thus, given
a forest $F$, and a positive integer $k$, the problem of deciding
whether $F$ can be transformed into a linear forest with at most
$k$ vertex deletions can be solved in linear time.
\end{theorem}

Moreover, based on the algorithm, we have the following results.

\begin{corollary}\label{cor:binary}$(\clubsuit)$
Let $T$ be a full binary tree with $n$ vertices, and $t =
\log_2(n+1) \! \mod 3$. Then $\cdn(T) = (n+1-2^t)/7.$
\end{corollary}

\begin{corollary}\label{cor:kAry}$(\clubsuit)$
Let $T$ be a full $k$-ary tree with $n$ vertices, for $k \geq 3$, and
$t = \log_k(nk-n+1) \! \mod 2$. Then $ \cdn(T) =
(nk-n+1-k^t)/(k^2-1).$
\end{corollary}

\section{Block Graphs}
\label{sec:blockGraphs}

We describe a dynamic programming algorithm to compute the minimum
weight of a claw-deletion set in a weighted connected block graph
$G$. The algorithm to be presented can be easily modified to
compute also a set realizing the minimum.

If the block graph $G$ has no cutpoint, the problem is trivial as
$G$ is already claw-free. Otherwise, let $T$ be the block-cutpoint-tree of the block graph $G$. Consider $T$ rooted at some
cutpoint $r$ of $G$, and let $v \in V(T)$. Let $G_v$ the subgraph of $G$
induced by the blocks in $T_v$. For $v \neq r$, let $G_v^+$ be the subgraph of $G$ induced by
the blocks in $T_v^+$. If $b$ is a block, let $G_b^- = G_b \setminus \{p_T(b)\}$ (notice that $p_T(b)$ is a cutpoint of $G$, and it is always defined because $r$ is not a block), and let $s(b)$ be the sum of weights of the vertices of $b$
that are not cutpoints of $G$ ($s(b) = 0$ if there is no such vertex).

We consider three functions to be computed for a vertex $v$ of $T$
that is a cutpoint of $G$:

\begin{itemize}

\item $f_1(v)$: the minimum weight of a claw-deletion set of $G_v$
containing $v$.

\item $f_2(v)$: the minimum weight of a claw-deletion set of $G_v$ not
containing $v$.

\item  For $v \neq r$, $f_3(v)$: the minimum weight of a claw-deletion set of $G_v^+$ containing neither $v$ nor all the vertices of
$p_T(v) \setminus \{v\}$ (notice that $p_T(v)$ is a block).

\end{itemize}

The parameter that solves the whole problem is $f(r) = \min\{f_1(r),f_2(r)\}$.

We define also three functions to be computed for a vertex $b$ of
$T$ that is a block of $G$:

\begin{itemize}

\item $f_1(b)$: the minimum weight of a claw-deletion set of $G_b^-$
containing $b \setminus \{p_T(b)\}$.

\item $f_2(b)$: the minimum weight of a claw-deletion set of $G_b^-$.

\item $f_3(b)$: the minimum weight of a claw-deletion set of $G_b$ not
containing $p_T(b)$.

\end{itemize}

We compute the functions in a bottom-up order as follows, where
$v$ (resp.~$b)$ denotes a vertex of $T$ that is a cutpoint (resp.
block) of $G$.
Notice that the leaves of $T$ are blocks of $G$.

If $C(b)=\emptyset$, then $f_1(b)=s(b)$, $f_2(b)=0$, and
$f_3(b)=0$. Otherwise,

\begin{itemize}\setlength\itemsep{4pt}
\item $f_1(v)=w(v)+\textstyle\sum_{b \in C(v)} f_2(b)$;
$f_1(b)=s(b) + \sum_{v \in C(b)} f_1(v)$;

\item if $|C(v)|\leq 2$, then $f_2(v)=\sum_{b \in C(v)} f_3(b)$;
if $|C(v)|\geq 3$, then $\textstyle f_2(v)=\min_{b_1,b_2 \in C(v)}
(\sum_{b \in \{b_1,b_2\}} f_3(b) + \sum_{b \in C(v)\setminus
\{b_1,b_2\}} f_1(b))$;

\item $\textstyle f_2(b)=\min \{ \sum_{v\in C(v)}
\min\{f_1(v),f_3(v)\},$ $\min_{v_1 \in C(v)} (s(b)+f_2(v_1)$
\linebreak $+~\sum_{v\in C(v)\setminus\{v_1\}} f_1(v))\}$;

\item $\textstyle f_3(b)=\sum_{v \in C(b)} \min\{f_1(v),f_3(v)\}$;

\item if $C(v) = \{b\}$, then $f_3(v)= f_3(b)$; \\ if $|C(v)|\geq
2$, then $\textstyle f_3(v)=\min_{b_1 \in C(v)} (f_3(b_1) +
\sum_{b \in C(v)\setminus \{b_1\}} f_1(b)).$
\end{itemize}

The explanation of the correctness of these formulas follows in Theorem~\ref{theo:block}.

\begin{theorem}\label{theo:block}$(\clubsuit)$
Let $G$ be a weighted connected block graph which is not complete.
Let $T$ be the block-cutpoint-tree of $G$, rooted at a cutpoint
$r$. The previous function $f(r)$ computes correctly the minimum
weight of a claw-deletion set of~$G$.
\end{theorem}

We obtain this result as a corollary.

\begin{corollary}\label{cor:timeBlockGraph}$(\clubsuit)$
Let $G$ be a weighted block graph with $n$ vertices and $m$ edges.
The minimum weight of a claw-deletion set of $G$ can be determined
in $\mathcal{O}(n+m)$ time.
\end{corollary}

\section{Graphs of Bounded Treewidth}
\label{sec:treewidth}

Next, we present an algorithm able of solving \textsc{Weighted
Claw-free Vertex Deletion} in linear time on graphs with bounded
treewidth, which also implies that we can recognize claw-free
graphs in linear time when the input graph has treewidth bounded
by a constant. For definitions of tree decompositions and
treewidth, we refer the reader
to~\cite{cygan2011solving,kloks-treewidth,R-S-minors3-pltw}.

%

%
Graphs of treewidth at most $k$ are called {\em partial
$k$-trees}. Some graph classes with bounded treewidth include:
forests (treewidth 1); pseudoforests, cacti, outerplanar graphs,
and series-parallel graphs (treewidth at most 2); Halin graphs and
Apollonian networks (treewidth at most
3)~\cite{bodlaender1998partial}. In addition, control flow graphs
arising in the compilation of structured programs also have
bounded treewidth (at most 6)~\cite{thorup1998all}.

Based on the following results we can assume that we are given a nice tree decomposition of the input graph $G$.

\begin{theorem}~\cite{bodlaender2016c}
There exists an algorithm that, given a $n$-vertex graph $G$ and an
integer $k$, runs in time $2^{\mathcal{O}(k)} \cdot n$ and either outputs that
the treewidth of $G$ is larger than $k$, or constructs a tree
decomposition of $G$ of width at most $5k + 4$.
\end{theorem}

\begin{lemma}~\cite{kloks-treewidth}
Given a tree decomposition $(T, \{X_t\}_{t\in V(T)})$ of $G$ of
width at most $k$, one can compute in time $\mathcal{O}(k^2\cdot
\max\{|V(T)|,|V(G)|\})$ a nice tree decomposition of $G$ of width
at most $k$ that has at most $\mathcal{O}(k \cdot |V(G)|)$ nodes.
\end{lemma}

Now we are ready to use a nice tree decomposition in order to obtain
a linear-time algorithm for \textsc{Weighted Claw-free Vertex Deletion} on graphs with bounded treewidth.

\begin{theorem}\label{theo:cfvdTreewidth}
\textsc{Weighted Claw-free Vertex Deletion} can be solved in linear time on graphs with bounded treewidth. More precisely, there is a $2^{\mathcal{O}(k^2)} \cdot n$-time algorithm to solve \textsc{Weighted Claw-free Vertex Deletion} on $n$-vertex graphs $G$ with treewidth at most $k$.
\end{theorem}

\begin{proof}
Let $G$ be a weighted $n$-vertex graph with $tw(G) \leq k$. Given
a nice tree decomposition $\mathcal{T} = (T, \{X_t\}_{t \in
V(T)})$ of $G$, we describe a procedure that computes the minimum
weight of a claw-deletion set of $G$ ($\cdn_w(G)$) using dynamic
programming. For a node $t$ of $T$, let $V_t = \bigcup_{t' \in
T_t} X_{t'}$. First, we will describe what should be stored in
order to index the table.
%
Given a claw-deletion set $\hat{S}$ of $G$, for any bag $X_t$ there is a partition of $X_t$ into $S_t,A_t,B_t$ and $C_t$ where
\begin{itemize}
\item $S_t$ is the set of vertices of $X_t$ that are going to be removed ($S_t=\hat{S}\cap X_t$);

\item $A_t = \{v \in X_t \setminus \hat{S}: |N_{V_t \setminus X_t}(v) \setminus \hat{S}| = 0 \}$ is the set of non-removed vertices of $X_t$ that are going to have no neighbor in $V_t\setminus X_t$ after the removal of $\hat{S}$;

\item $B_t=\{v \in X_t \setminus \hat{S}: N_{V_t \setminus X_t}(v) \setminus \hat{S}\mbox{ induces a non-empty clique} \}$ is the set of non-removed vertices of $X_t$ that, after the removal of $\hat{S}$, are going to have neighbors in $V_t\setminus X_t$, but no pair of non-adjacent neighbors;

\item $C_t=\{v \in X_t \setminus \hat{S}: \text{ there exist } u,u' \in N_{V_t \setminus X_t}(v)\setminus \hat{S} \text{ with } uu' \notin E(G) \}$
is the set of non-removed vertices of $X_t$ that, after the removal of $\hat{S}$, are going to have a pair of non-adjacent neighbors in $V_t\setminus X_t$.
\end{itemize}

In addition, the claw-deletion set $\hat{S}$ also provides the set
$Z_t=\{ (x,y) \in (X_t \setminus \hat{S}) \times (X_t \setminus
\hat{S}) : \exists ~w \in V_t \setminus (X_t \cup \hat{S})$ with $
xy, wy \in E(G)$ and $wx \notin E(G)\}$ which consists of ordered
pairs of vertices $x,y$ of $X_t$ that, after the removal of
$\hat{S}$, are going to induce a $P_3=x,y,w$ with some $w\in
V_t\setminus X_t$.

Therefore, the recurrence relation of our dynamic programming has
the signature $\cdn_w[t,S,A,B,C,Z]$, representing the minimum
weight of a vertex set whose removal from $G[V_t]$ leaves a
claw-free graph, such that $S,A,B,C$ form a partition of $X_t$ as
previously described, and $Z$ is as previously described too. The
generated table has size $2^{\mathcal{O}(k^2)} \cdot n$.

Function $\cdn_w$ is computed for every node $t \in V(T)$, for
every partition $S \cup A \cup B \cup C$ of $X_t$, and for every
$Z \subseteq X_t \times X_t$. The algorithm performs the
computations in a bottom-up manner. Let $T$ rooted at $r \in
V(T)$. Notice that $V_r = V(G)$, then $\cdn_w[r,
\emptyset,\emptyset,\emptyset,\emptyset,\emptyset]$ is the weight
of a minimum weight claw-deletion set of $G_r=G$, which solves the
whole problem.

We present additional terminology.
Let $t$ be a node in $T$ with children $t'$ and $t''$, and $X \subseteq X_t$.
To specify the sets $S, A, B, C$ and $Z$ on $t'$ and $t''$, we employ the notation $S', A', B', C', Z'$ and $S'', A'', B'', C'', Z''$, respectively.

Now, we describe the recurrence formulas for the function $\cdn_w$
defined, based on the types of nodes in $T$.

\begin{itemize} \setlength{\parskip}{0pt} \setlength{\itemsep}{2pt plus 0pt}
\item \textbf{Leaf node.} If $t$ is a leaf node in $T$, then
    $\cdn_w[t, \emptyset, \emptyset, \emptyset, \emptyset, \emptyset] = 0.$ \hfill $(1)$

\item \textbf{Introduce node.} Let $t$ be an introduce node with child $t'$ such that $X_t = X_{t'} \cup \{v\}$ for some vertex $v \notin X_{t'}$. Let $S \cup A \cup B \cup C$ be a partition of $X_t$, and $Z \subseteq X_t \times X_t$.
The recurrence is given by the following formulas.
\begin{itemize} \setlength{\parskip}{0pt} \setlength{\itemsep}{2pt plus 0pt}
\item If $v \in S$, then \\ $\cdn_w[t,S,A,B,C,Z]  = \cdn_w[t',S
\setminus \{v\}, A,B,C,Z]+w(v).$ \hfill $(2.1)$


\item If $v \in A$, then
$\cdn_w[t,S,A,B,C,Z]  =\cdn_w[t', S, A \setminus \{v\}, B, C, Z'],$ \hfill $(2.2)$\\
 if  $N_{X_t \setminus S}(v) \text{ does not induce a } \overline{K}_3,$
 for every $(x,y) \in Z, vx\in E(G)$ or $vy \notin E(G),$
 $N_{X_t}(v) \cap C = \emptyset,$ there is $Z'$ such that
 $Z = Z' \cup \{(v,y): y\in B\cup C\mbox{ and } vy\in E(G)\}.$\\
 Otherwise, $\cdn_w[t,S,A,B,C,Z]  = \infty.$

\item If $v \in B\cup C$, then $ \cdn_w[t,S,A,B,C,Z]  = \infty.$
\hfill $(2.3)$

\end{itemize}

\item \textbf{Forget node.} Consider $t$ a forget node with child $t'$ such that $X_t = X_{t'} \setminus \{v\}$ for some vertex $v \in X_{t'}$. Let $S \cup A \cup B \cup C$ be a partition of $X_t$, and $Z \subseteq X_t \times X_t$.

If $N_A(v) \neq \emptyset$, then $\cdn_w[t,S,A,B,C,Z]  =
\cdn_w[t', S \cup \{v\}, A,B,C,Z].$  \hfill $(3.1)$

Otherwise, $\cdn_w[t,S,A,B,C,Z]  =$ \\ $\min \big\{ \cdn_w[t', S
\cup \{v\}, A, B, C, Z], \cdn_w[t', S, A', B', C', Z']\big\},$
\hfill $(3.2)$

 among every $(S,A',B',C',Z')$ such that: \\
 $Z = (Z' \setminus \{(x,y) : x = v \text{ or } y = v \}) \cup
  \{(x,y) \in  X_t \times X_t : xy,vy \in E(G) \text{ and } vx \notin E(G) \},\\
 A = A' \setminus N_G[v], \,
 B = ((B' \setminus \{b \in B' : (v,b) \in Z'\}) \cup (A' \cap N_{G}(v))) \setminus \{v\},\\
 C = (C' \cup  \{b \in B' : (v,b) \in Z'\}) \setminus \{v\}.$

\item \textbf{Join node.} Consider $t$ a join node with children $t',t''$ such that $X_t = X_{t'} = X_{t''}$. Let $S \cup A \cup B \cup C$ be a partition of $X_t$, and $Z \subseteq X_t \times X_t$.
The recursive formula is given by

$\cdn_w[t,S,A,B,C,Z] =$ \\ $\min\big\{ \cdn_w[t',S',A',B',C',Z'] +
\cdn_w[t'',S'',A'',B'',C'',Z''] \big\} - w(S),$\hfill $(4)$
among every $(S',A',B',C',Z')$ and $(S'',A'',B'',C'',Z'')$ such
that:
$S = S' = S''$;
 $A = A' \cap A''$;
 $B = (A' \cap B'') \cup (A'' \cap B')$;
 $C = C' \cup C'' \cup (B' \cap B'')$;
 $Z = Z' \cup Z''.$
\end{itemize}

We explain the correctness of these formulas. The base case is
when $t$ is a leaf node. In this case $X_t = \emptyset$, then all
the sets $S,A,B,C,Z$ are empty. The set $X_t = \emptyset$ also
implies that $G[V_t]$ is the empty graph, which is claw-free.
Hence, $\cdn_w(G[V_t])=0$ and Formula (1) holds.

Let $t$ be an introduce node with child $t'$, and $v$ the vertex
introduced at $t$. First, suppose that $v \in S$. We assume by
inductive hypothesis that $G[V_{t'} \setminus \hat{S}]$ is
claw-free. Since $v \in S \subseteq \hat{S}$, we obtain that
$G[V_t \setminus (\hat{S} \cup \{v\})]$ is claw-free. Then, the
weight of a minimum weight claw-deletion set of $G[V_t]$ is
increased by $w(v)$ from the one of $G[V_{t'}]$, stored at
$\cdn_w[t',S',A',B',C',Z']$. Since $v \in S$, then $v \notin S'$
and the sets $A',B',C',Z'$ in node $t'$ are the same $A,B,C,Z$ of
$t$. Consequently Formula (2.1) holds.

Now, suppose that $v \in A \cup B \cup C$. By definition of tree
decomposition, $v \notin N_{V_t \setminus X_t}(X_t)$. Then, if
$v \in B \cup C$, the partition $S \cup A \cup B \cup C$ is not
defined as required, and this justifies Formula (2.3). Thus, let
$v \in A$. We have three cases in which $G[V_t \setminus \hat{S}]$
contains an induced claw:  ($i$) $N_{X_t}(v)$ induces a
$\overline{K}_3$, or ($ii$) there exists $(x,y) \in Z$, such that
$vx \notin E(G)$ and $vy \in E(G)$, or ($iii$) there exists $c \in
C$ such that $cv \in E(G)$. A set $Z$ according to definition of
$\cdn_w$ is obtained by $Z'$ together with the pairs $(x,y)$ such
that $x = v$, $xy \in E(G)$ and $y$ has at least one neighbor in
$V_t \setminus (X_t \cup \hat{S})$. (Note that $v = y$ is never
achieved, since $v$ is an introduce node and $v \notin N_{V_t
\setminus X_t}(X_t)$). Then, $Z = Z' \cup \{(v,y) : y \in B \cup
C \text{ and } vy\in E(G)\}$.
Hence, Formula (2.2) is justified by the negation of each of cases ($i$), ($ii$), ($iii$).

Next, let $t$ be a forget node with child $t'$. Let $v$ be the vertex forgotten at $t$. We consider $N_{A}(v) \neq \emptyset$ or not.
Notice that if $N_{G}(v) \cap A \neq \emptyset$ and $v \notin \hat{S}$, then we have a contradiction to the definition of $A$, because some $a \in A$ is going to have a neighbor in $V_t \setminus (X_t \cup \hat{S})$. 
Therefore, if $N_A(v) \neq \emptyset$,  $v$ indeed must belong to $\hat{S}$, then Formula~(3.1) holds.

Otherwise, consider that $N_{A}(v) = \emptyset$. In this
case, either $v \in \hat{S}$ or $v \notin \hat{S}$. Then, we
choose the minimum between these two possibilities. If $v \in
\hat{S}$ we obtain the value stored at $\cdn_w[t', S \cup \{v\},
A, B, C, Z]$. Otherwise, let $v \notin \hat{S}$. It follows that,
for some $a \in A$, if $va \in E(G)$, then $a$ must now belong to $B$.
Consequently, $A$ must be $A' \setminus N_G[v]$. Let $\mathcal{B}
= \{b \in B' : (v,b) \in Z'\}$. Since $v \notin \hat{S}$, for
every $x \in \mathcal{B}$, $x$ must belong to $C$. Thus, the set
$B$ is given by $B' \setminus \mathcal{B}$ together with the
vertices from $A'$ that now belong to $B$. Recall that $v \notin
X_t$, then $v \notin B$. Hence, $B = ((B' \setminus \mathcal{B})
\cup (A' \cap N_{G}(v)) ) \setminus \{v\}$. Finally, $C = (C' \cup
\mathcal{B}) \setminus \{v\}$. Hence, Formula (3.2) holds.

To conclude, let $t$ be a join node with children $t'$ and $t''$.
Note that the graphs induced by $V_{t'}$ and by $V_{t''}$ can be
distinct. Then, we must sum the values of $\cdn_w$ in $t'$ and in
$t''$ to obtain $\cdn_w$ in $t$, and choose the minimum of all of
these possible sums. Finally, we subtract $w(S)$ from the previous
result, since $w(S)$ is counted twice.

By definition of join node, $X_{t} =X_{t'} = X_{t''}$, then $S = S' = S''$.
Let $x \in X_t$. We have that $x \in A$ if and only if $|N_{V_{t'} \setminus X_{t'}}(x) \setminus \hat{S}| = |N_{V_{t''} \setminus X_{t''}}(x) \setminus \hat{S}| = 0$. Then, $A = A' \cap A''$.

Notice that $x \in B$ if and only if ($|N_{V_{t'} \setminus X_{t'}}(v) \setminus \hat{S}| = 0$ and $|N_{V_{t''} \setminus X_{t''}}(v) \setminus \hat{S}| > 1$) or ($|N_{V_{t''} \setminus X_{t''}}(v) \setminus \hat{S}| = 0$ and $|N_{V_{t'} \setminus X_{t'}}(v) \setminus \hat{S}| > 1$). Consequently $x \in B$ if and only if $x \in (A' \cap B'') \cup (A'' \cap B')$. This implies that $B = (A' \cap B'') \cup (A'' \cap B')$.

Now, $x \in C$ if and only if $x \in C'$ or $x \in C''$ or ($x \in B'$ and $x \in B''$). (Note that by the definition of tree decomposition, the forgotten nodes in $G_{t'}$ and $G_{t''}$ are distinct and therefore the condition $x \in B'$ and $x \in B''$ is safe).
Consequently, $C =  C' \cup C'' \cup (B' \cap B'')$.

Finally, let $x,y \in X_t$. By definition of $Z'$, if $(x,y) \in Z'$, then there exists $w \in V_{t'} \setminus (X_{t'} \cup \hat{S})$  with $xy, wy \in E(G) \text{ and } wx \notin E(G)$. This implies that $w \in V_{t} \setminus (X_{t} \cup \hat{S})$ and $xy, wy \in E(G) \text{ and } wx \notin E(G)$. Hence, $(x,y) \in Z$. By a similar argument, we conclude that if $(x,y) \in Z''$, then $(x,y) \in Z$.
This gives $Z = Z' \cup Z''$, and completes Formula (4).


Since the time to compute each entry of the table is upper bounded
by $2^{\mathcal{O}(k^2)}$ (see Appendix) and the table has size
$2^{\mathcal{O}(k^2)} \cdot n$, the algorithm can be performed in
$2^{\mathcal{O}(k^2)} \cdot n$ time. This implies linear-time
solvability for graphs with bounded treewidth. \qed\end{proof}
%
%
%
\bibliographystyle{splncs04}
\bibliography{bnm-j,bnm}

\newpage

\appendix

\section*{Appendix}

\paragraph*{\bf Some Definitions.}

Let $G$ be a graph. Given a vertex $v \in V(G)$, its \textit{open neighborhood} consists of all adjacent vertices to $v$ and is denoted by $N_G(v)$, whereas its \textit{closed neighborhood} is the set $N_G[v] = N_G(v) \cup \{v\}$.  For a set $U \subseteq V(G)$, let $N_G(U) = \bigcup_{v \in U} N_G(v) \setminus U$, and $N_G[U] = N_G(U) \cup U$. When the graph $G$ is clear from the context, we denote $N_G(v) \cap U$ by $N_U(v)$.

The \textit{degree} of a vertex $v \in V(G)$ on a set $U \subseteq V(G)$, is $d_U(v) = |N_G(u) \cap U|$. If $U = V(G)$, we simply write $d_G(u)$.
We say that $v \in V(G)$ is an \textit{isolated} (resp. a \textit{leaf}) vertex if $d_G(v) = 0$ (resp. $d_G(v) = 1$).
A set $U \subseteq V(G)$ is called a \textit{clique} if the vertices in $U$ are pairwise adjacent.

For $U \subseteq V(G)$, the subgraph of $G$ \textit{induced} by $U$, denoted by $G[U]$, is the graph whose vertex set is $U$ and whose edge set consists of all the edges in $E(G)$ that have both endpoints in $U$.  If $H$ is a subgraph of $G$, we write $H \subseteq G$.
For $U \subseteq V(G)$, we denote by $G \setminus U$ the graph $G[V(G) \setminus U]$.

A graph is \textit{connected} is every pair of vertices is joined by a path. A maximal connected subgraph of $G$ is called a \textit{connected component} of $G$.
A graph $G$ is called $k$-\textit{connected} if $G \setminus X$ is connected for every set $X \subseteq V(G)$ with $|X| \leq k$.
A \textit{block} of a graph $G$ is a maximal $2$-connected subgraph of $G$. A vertex $v$ of a graph $G$ is a \textit{cutpoint} if $G \setminus \{v\}$ has more connected components than $G$.

A \textit{block graph} is a graph in which every block is a clique. A \textit{forest} is an acyclic graph or, equivalently, a graph in which every block is an edge. A \textit{linear forest} is the disjoint union of induced paths. A \textit{tree} is a connected forest.

A $k$\textit{-ary tree} is a rooted tree $T$ in which every node of $T$ has at most $k$ children. In particular, for $k = 2$, and $k = 3$ we have the \textit{binary}, and the \textit{ternary} tree, respectively.
A \textit{strict $k$-ary tree} is a rooted tree $T$ in which every node of $T$ has either zero or $k$ children.
The \textit{depth} of a vertex $v \in V(T)$ is the length of a path from $v$ to $r$ in $T$. A \textit{full $k$-ary tree}  is a strict $k$-ary tree in which all leaves have the same depth. 

A graph $G$ is a \textit{split graph} if $V(G)$ admits a partition $V(G) = C \cup I$ into a clique $C$ and an independent set $I$. A graph is \textit{chordal} if every cycle of length greater than three has a \textit{chord}, i.e., an edge between two non-consecutive vertices of the cycle.
Forests, block graphs, and split graphs are all subclasses of chordal graphs.




\begin{definition}\cite{R-S-minors3-pltw}
A \emph{tree decomposition} of a graph $G$ is a pair $\mathcal{T} = (T, \{X_t\}_{t \in V(T)})$, where $T$ is a tree whose every node $t$ is assigned a vertex subset $X_t \subseteq V(G)$ called \emph{bag}, such that the following three conditions hold:
\begin{itemize}
\item $\bigcup_{t \in V(T)} X_t = V(G)$.
\item For every $uv \in E(G)$, there exists a node $t$ of $T$ such that bag $X_t$ contains both $u$ and $v$.
\item For every $u \in V(G)$, the set $T_u = \{t \in V(T) : u \in X_t\}$ induces a connected subgraph of $T$.
\end{itemize}
\end{definition}

The \emph{width} of a tree decomposition is $\max_{t \in V(T)} (|X_t| - 1)$. The \emph{treewidth} $tw(G)$ of a graph $G$ is the minimum possible width of a tree decomposition of $G$.

\begin{definition} \cite{kloks-treewidth}
A \emph{nice tree decomposition} is a tree decomposition with one
special node $r$ called \emph{root} with $X_r = \emptyset$, and
each node is one of the following types:
\begin{itemize}
\item \textbf{Leaf node:} a leaf $\ell$ of $T$ with $X_\ell = \emptyset$.
\item \textbf{Introduce node:} a node $t$ with exactly one child $t'$ such that $X_t = X_{t'} \cup \{v\}$ for some vertex $v \notin X_{t'}$; we say that $v$ is \emph{introduced} at $t$.
\item \textbf{Forget node:} a node $t$ with exactly one child $t'$ such that $X_t = X_{t'} \setminus \{v\}$ for some vertex $v \in X_{t'}$; we say that $v$ is \emph{forgotten} at $t$.
\item \textbf{Join node:} a node $t$ with two children $t',t'$ such that $X_t = X_{t'} = X_{t'}$.
\end{itemize}
\end{definition}

\paragraph*{\bf Proof of Theorem~\ref{theo:clawLinearForestJ}.}
We will first prove that Algorithm~\ref{alg:forestClawDeletionSet}
is correct for trees and that it runs in linear time. Then we will
generalize the result to forests.

\begin{theorem}\label{theo:clawLinearForest}
Let $T$ be a rooted tree of order $n$. A minimum claw-deletion set
of $T$ can be found by Algorithm~\ref{alg:forestClawDeletionSet}
in $\mathcal{O}(n)$ time.
\end{theorem}

\begin{proof}
Let $T$ be a rooted tree. We will prove by induction that
Algorithm~\ref{alg:forestClawDeletionSet} is correct. The basis is
the case when $C(v)=\emptyset$. Since $T_v^+$ consists either of a
single edge or of a single vertex (when $p = \emptyset$), clearly
the empty set is a minimum claw-deletion set of $T_v^+$. Moreover,
$\cdn(T_v^+)=\cdn(T_v)=\cdn(T_v\setminus \{v\})$. Hence, function
\textsc{Claw-Deletion-Set} is correct when $C(v) = \emptyset$.

Suppose that $C(v)\neq\emptyset$ and let $C(v) =
\{u_1,\dots,u_k\}$, for $k \geq 1$. For the inductive hypothesis,
we assume that $S_i = \textsc{Claw-Deletion-Set}(T,u_i,v)$ is a
minimum claw-deletion set of $T_{u_i}^+$, for every $1 \leq i \leq
k$, such that: if $\cdn(T_{u_i}^+)=1+\cdn(T_{u_i})$ then $v \in
S_i$; if $\cdn(T_{u_i}^+)=\cdn(T_{u_i})$ and
$\cdn(T_{u_i})=1+\cdn(T_{u_i}\setminus \{u_i\})$ then $u_i \in
S_i$.

Let $S = S_1 \cup \dots \cup S_k$.

If $v \in S$, then the connected components of $T_v^+ \setminus S$
are $\{p\}$ (when $p \neq \emptyset$) and the connected components
of $T_{u_i}^+ \setminus S$, for $1 \leq i \leq k$, which, by
inductive hypothesis, are induced paths. So $S$ is a claw-deletion
set of $T_v^+$. Also, by minimality, $S_i \setminus \{v\}$ is a
minimum claw-deletion set of $T_{u_i}$, for every $1 \leq i \leq
k$. Let $1 \leq j \leq k$ such that $v \in S_j$. Then $\cdn(T_v^+)
\geq \cdn(T_v) \geq \cdn(T_{u_j}^+) + \sum_{1 \leq i \leq k; i\neq
j} \cdn(T_{u_i}) = |S|$. Thus, $S$ is a minimum claw-deletion set
of $T_v^+$ and $\cdn(T_v^+) = \cdn(T_v)$. This also implies that
$S$ satisfies the further conditions required to the output.

From now on, suppose that $v \not \in S$. Then, by inductive
hypothesis, $\cdn(T_{u_i}^+)=\cdn(T_{u_i})$ and, moreover, $S_i$
is also a minimum claw-deletion set of $T_{u_i}$, for every $1
\leq i \leq k$. Let $c = |C(v) \setminus S|$. For each $u_i \in
C(v) \setminus S$, it also holds
$\cdn(T_{u_i})=\cdn(T_{u_i}\setminus \{u_i\})$ and $S_i$ is a
minimum claw-deletion set of $T_{u_i} \setminus \{u_i\}$

Suppose first that $c \leq 1$, i.e., $C(v) \setminus S \subseteq
\{u_j\}$ for some $1 \leq j \leq k$. Then, the connected
components of $T_v^+ \setminus S$ are the connected components of
$T_{u_i} \setminus S$, for $1 \leq i \leq k$, $i \neq j$, plus the
connected components of $T_{u_j}^+ \setminus S$ which, by
inductive hypothesis, are induced paths, with the addition of
vertex $p$ (when $p \neq \emptyset$) to the path containing $v$.
It is easy to see that the resulting component is still an induced
path. So $S$ is a claw-deletion set of $T_v^+$. Since $\cdn(T_v^+)
\geq \cdn(T_v) \geq \sum_{1 \leq i \leq k} \cdn(T_{u_i}) = |S|$,
$S$ is a minimum claw-deletion set of $T_v^+$ and $\cdn(T_v^+) =
\cdn(T_v) = \cdn(T_v \setminus \{v\})$. This also implies that $S$
satisfies the further conditions required to the output.

Suppose now that $c \geq 3$. Using the inductive hypothesis and
similarly to the case where $v \in S$, it is not difficult to see
that $S \cup \{v\}$ is a claw-deletion set of $T_v^+$. Moreover,
$\cdn(T_v^+) \geq \cdn(T_v) \geq \cdn(T[\{v\} \cup C(v) \setminus
S]) + \sum_{u_i \in C(v) \setminus S} \cdn(T_{u_i} \setminus
\{u_i\}) + \sum_{u_i \in C(v) \cap S} \cdn(T_{u_i}) = 1+|S|$. This
shows that $S \cup \{v\}$ is a minimum claw-deletion set of
$T_v^+$ and $\cdn(T_v^+) = \cdn(T_v)$. So $S \cup \{v\}$ satisfies
the required conditions.

Finally, suppose that $c = 2$, i.e., $C(v) \setminus S = \{u_j,
u_{j'}\}$ for some $1 \leq j < j' \leq k$. The connected
components of $T_v \setminus S$ are the connected components of
$T_{u_i} \setminus S$, for $1 \leq i \leq k$, $i \neq j, j'$ plus
the connected components of $T_{u_j}^+ \setminus S$ and of
$T_{u_{j'}}^+ \setminus S$ not containing $v$ which, by inductive
hypothesis, are induced paths, plus a path having $u_j v u_{j'}$
as a subpath. So $S$ is a claw-deletion set of $T_v$ and $S \cup
\{p\}$ is a claw-deletion set of $T_v^+$ when $p \neq \emptyset$.
Notice that, when $p \neq \emptyset$, $T[\{p,v,u_j,u_{j'}\}]$ is a
claw, so $\cdn(T[\{p,v,u_j,u_{j'}\}])=1$. When $p = \emptyset$,
$\cdn(T[\{p,v,u_j,u_{j'}\}])=0$. Then $\cdn(T_v^+) \geq
\cdn(T[\{p,v,u_j,u_{j'}\}]) + \cdn(T_{u_j} \setminus \{u_j\}) +
\cdn(T_{u_{j'}} \setminus \{u_{j'}\}) + \sum_{1 \leq i \leq k; i
\neq j, j'} \cdn(T_{u_i}) = 1+|S|$ when $p \neq \emptyset$, and
$|S|$ otherwise. This shows that $S \cup \{p\}$ (resp. $S$) is a
minimum claw-deletion set of $T_v^+$ when $p \neq \emptyset$
(resp. when $p = \emptyset$). In the first case, by minimality,
$S$ is also a minimum claw-deletion set of $T_v$, so $\cdn(T_v^+)
= 1+\cdn(T_v)$, and $S \cup \{p\}$ satisfies the required
conditions. In the second case, $\cdn(T_v)=\cdn(T_v\setminus
\{v\})$, so $S$ satisfies the required conditions.

Therefore, \textsc{Claw-Deletion-Set} returns correctly a minimum
claw-deletion set of $T_v^+$ satisfying that if
$\cdn(T_v^+)=1+\cdn(T_v)$ then $p \in S$, and if
$\cdn(T_v^+)=\cdn(T_v)$ and $\cdn(T_v)=1+\cdn(T_v\setminus \{v\})$
then $v \in S$.

\smallskip

Next, we perform the runtime analysis of Algorithm~\ref{alg:forestClawDeletionSet}.

First, we have that checking each conditional statement of
Algorithm~\ref{alg:forestClawDeletionSet} requires
$\mathcal{O}(1)$ time if the tree is represented by lists of
children. Initializing $S=\emptyset$ at the very beginning of the
algorithm can be done in $\mathcal{O}(n)$ time by representing $S$
by an array. In that case, adding a vertex to $S$ can be done in
constant time. The assignment and union operations of Line~6 of
the algorithm are not necessary if all the recursive calls work on
the same array representing the set $S$. Line~7 computes the
number of children of a vertex $v$ which are not in $S$. Having
the list of children and $S$ represented by an array, this step
takes $\mathcal{O}(d_T(v))$ time. Since function
\textsc{Claw-Deletion-Set} is executed exactly one time for every
vertex $v \in V(T)$, we conclude that
Algorithm~\ref{alg:forestClawDeletionSet} runs in
$\mathcal{O}(n+m) = \mathcal{O}(n)$ time. \qed\end{proof}

From Theorem~\ref{theo:clawLinearForest}, we obtain the following
Corollary~\ref{cor:linearForest}, and together imply
Theorem~\ref{theo:clawLinearForestJ}.

\begin{corollary}\label{cor:linearForest}
Given a forest $F$, and a positive integer $k$, the problem of
deciding whether $F$ can be transformed into a linear forest with
at most $k$ vertex deletions can be solved in linear time.
\end{corollary}

\paragraph*{\bf Exact Values for Full $\mathbf{k}$-ary Trees.}

We determine the claw-deletion number of a $k$-ary tree $T$ with
height $h$, as a function of $k$ and $h$. The cases $k = 2$ and $k
\geq 3$ follow in Theorems~\ref{theo:binary}~and~\ref{theo:kAry},
respectively.

\begin{theorem}\label{theo:binary} 
Let $T$ be a full binary tree of height $h$, and $t = (h+1) \!
\mod 3$. Then $ \cdn(T) = (2^{h+1}-2^t)/7. $
\end{theorem}

\begin{proof}
Algorithm~\ref{alg:forestClawDeletionSet} chooses a claw-deletion
$S$ of $T$ comprised by all the vertices in depth $h-2$.
Subsequently, the same procedure chooses all the vertices in depth
$h-5$, and so on, until the depth $t  = (h+1) \! \mod 3$. For
every $1 \leq i \leq k$, the amount of vertices in depth $i$ is
$2^i$. Then
\[ \cdn(T) = |S| = 2^{h-2} + 2^{h-5} + \dots + 2^t. \]
That leads to a geometric progression with ratio $r = 2^{-3}$, and
$(h-t+1)/3$ terms, which results in $\cdn(T) = (2^{h+1}-2^t)/7$.
\qed\end{proof}

The result of Theorem~\ref{theo:binary} can be rewritten as a
function of the order of $T$.

\paragraph*{\bf Proof of Corollary~\ref{cor:binary}.} \emph{Let $T$
be a full binary tree with $n$ vertices, and $t = \log_2(n+1) \!
\mod 3$. Then $\cdn(T) = \displaystyle \frac{n+1-2^t}{7}.$}

\begin{proof}
We know that a full binary tree with $n$ vertices has height $h =
\log_2(n+1)-1$. By Theorem~\ref{theo:binary} with $t = \log_2(n+1)
\! \mod 3$, we obtain
\[\cdn(T) = \frac{2^{h+1}-2^t}{7} = \frac{2^{\log_2(n+1)}-2^t}{7} = \frac{n+1-2^t}{7}.\]
\qed\end{proof}

Next, we proceed to full $k$-ary trees with $k \geq 3$ in
Theorem~\ref{theo:kAry}.

\begin{theorem}\label{theo:kAry}
Let $T$ be a full $k$-ary tree of height $h$, for $k \geq 3$, and
$t = (h-1) \! \mod 2$. Then $\cdn(T) = (k^{h+1}-k^t)/(k^2-1).$
\end{theorem}

\begin{proof}
Algorithm~\ref{alg:forestClawDeletionSet} chooses a claw-deletion
$S$ of $T$ comprised by all the vertices in depth $h-1$, all the
vertices in depth $h-3$, and so on, until the depth $t  = (h-1) \!
\mod 2$. Then,
\[\cdn(T) = |S| = k^{h-1} + k^{h-3} + \dots + k^t.\]
That leads to a geometric progression with ratio $r = k^{-2}$, and
$(h-t+1)/2$ terms, which follows that $\cdn(T) =
(k^{h+1}-k^t)/(k^2-1)$. \qed\end{proof}

Theorem~\ref{theo:kAry} rewritten as a function of the order of
$T$ follows below.

\paragraph*{\bf Proof of Corollary~\ref{cor:kAry}.} \emph{Let $T$ be
a full $k$-ary tree with $n$ vertices, for $k \geq 3$, and $t =
\log_k(nk-n+1) \! \mod 2$. Then $ \cdn(T) = \displaystyle
\frac{nk-n+1-k^t}{k^2-1}. $ }

\begin{proof}
We know that a full $k$-ary tree with $n$ vertices has height $h =
\log_k(nk-n+1)-1$. By Theorem~\ref{theo:binary} with $t =
\log_k(nk-n+1) \! \mod 2$, we obtain
\[ \cdn(T) = \frac{k^{h+1}-k^t}{k^2-1} = \frac{k^{\log_k(nk-n+1)}-k^t}{k^2-1} = \frac{nk-n+1-k^t}{k^2-1}.\]
\qed\end{proof}

We establish the proportion of vertices in $V(T)$ that belongs to
a claw-deletion set of $T$.

\begin{corollary}
Let $T$ be a full $k$-ary tree of order $n$ and height $h$. Let $t
= (h+1) \! \mod 3$ and $t' = (h-1) \! \mod 2$. It holds that
 \[ \frac{\cdn(T)}{n} =
    \left \{
    \begin{array}{l l}
      \frac{2^{h+1}-2^t}{7(2^{h+1}-1)}, & \text{if $k = 2$;} \\
      \frac{k^{h+1}-k^{t'}}{(k+1)(k^{h+1}-1)}, & \text{if $k \geq 3$.}
    \end{array}
    \right.\]
 In addition, $t = t' = 0$ implies
 \[ \frac{\cdn(T)}{n} =
    \left \{
    \begin{array}{l l}
      1/7, & \text{if $k = 2$;} \\
      1/(k+1), & \text{if $k \geq 3$.}
    \end{array}
    \right.\]
\end{corollary}

Among  full $k$-ary trees, $k = 3$ maximizes the proportion of
vertices in a claw-deletion set.

\paragraph*{\bf Proof of Theorem~\ref{theo:block}.} \emph{
Let $G$ be a weighted connected block graph which is not complete.
Let $T$ be the block-cutpoint-tree of $G$, rooted at a cutpoint
$r$. The previous function $f(r)$ computes correctly the minimum
weight of a claw-deletion set of $G$.}

\begin{proof}
We will prove by induction (bottom-up), that $f_1$, $f_2$, $f_3$
on $V(T)$ correctly compute the weight stated in their definition.
In that case, being $r$ a cutpoint of $G$ and $G_r = G$, it is
clear that $f(r)$ computes the minimum weight of a claw-deletion
set of $G$.

Let $b$ be a leaf of $T$. Then $b$ is a block of $G$, and $G_b$
and $G_b^-$ are complete, so any set is a claw-deletion set of
$G_b$ and $G_b^-$. Moreover, every vertex of $b \setminus
\{p_T(b)\}$ is simplicial in $G$, so the weight of $b \setminus
\{p_T(b)\}$ is $s(b)$. Thus, $f_1(b) = s(b)$, $f_2(b) = f_3(b) =
0$ is correct.

Now, let $v$ be a cutpoint of $G$ and, by inductive hypothesis,
assume that for the children of $v$ in $T$ the values of $f_1$,
$f_2$, and $f_3$ are correct according to their definition.

Consider first $f_1(v)$, i.e., the minimum weight of a
claw-deletion set of $G_v$ containing $v$. The connected
components of $G_v \setminus \{v\}$ are $\{G_b^-\}_{b\in C(v)}$.
So, it is enough to compute the minimum weight of a claw-deletion
set of each of them, and add to their sum the weight of $v$, so
$f_1(v)=w(v)+\sum_{b \in C(v)} f_2(b)$.

Consider next $f_2(v)$, i.e., the minimum weight of a
claw-deletion set of $G_v$ not containing $v$. In this case, we
have to avoid claws having $v$ as a center and the tree leaves in
three distinct blocks of $C(v)$, so all but at most two of the
blocks have to be completely contained in the set, except for
vertex $v$. For the remaining (at most two) blocks $b$, we need to
compute the minimum weight of a claw-deletion set of $G_b$ not
containing $v$ (which is $p_T(b)$). This justifies the formula
$f_2(v)=\sum_{b \in C(v)} f_3(b)$ for $|C(v)|\leq 2$, and
$f_2(v)=\min_{b_1,b_2 \in C(v)} (\sum_{b \in \{b_1,b_2\}} f_3(b) +
\sum_{b \in C(v)\setminus \{b_1,b_2\}} f_1(b))$, otherwise.

Finally, consider $f_3(v)$, for $v \neq r$, i.e., the minimum
weight of a claw-deletion set of $G_v^+$ containing neither $v$
nor all the vertices of $p_T(v) \setminus \{v\}$ (recall that
$p_T(v)$ is a block). In this case, we have to avoid claws having
$v$ as a center, one leaf in $p_T(v) \setminus \{v\}$, and two
other leaves in two distinct blocks of $C(v)$. So all but at most
one of the blocks have to be completely contained in the set,
except for vertex $v$. For the remaining block $b$, we need to
compute the minimum weight of a claw-deletion set of $G_b$ not
containing $v$ (which is $p_T(b)$). This justifies the formula
$f_3(v)= f_3(b)$ when $C(v) = \{b\}$, and $f_3(v)=\min_{b_1 \in
C(v)} (f_3(b_1) + \sum_{b \in C(v)\setminus \{b_1\}} f_1(b))$,
otherwise.

To conclude the proof, let $b$ be a node which is a block of $G$
and, by inductive hypothesis, assume that for the children of $b$
the values of $f_1$, $f_2$, and $f_3$ are correct according to
their definition.

Consider first $f_1(b)$, i.e., the minimum weight of a
claw-deletion set of $G_b^-$ containing all the vertices of $b
\setminus \{p_T(b)\}$. All the claws containing vertices of $b$
are hit by the set by definition, so it is enough to compute for
every $v$ in $C(b)$ the minimum weight of a claw-deletion set of
$G_v$ containing $v$, and adding to it the weight of all the
simplicial vertices of $b$, that is, $s(b)$. Then the formula
$f_1(b)=s(b) + \sum_{v \in C(b)} f_1(v)$ is correct.

Consider next $f_3(b)$, the minimum weight of a claw-deletion set
of $G_b$ not containing $p_T(b)$. For each $v$ in $C(b)$, either
$v$ belongs to the set, or $v$ does not belong to the set and
there is another vertex of $b$ that does not belong to the set.
So, we have to recursively compute $f_1(v)$ or $f_3(v)$,
respectively, and choose the minimum. In this case the simplicial
vertices do not belong to the minimum weight set, since the
weights are positive. This justifies the formula $f_3(b)=\sum_{v
\in C(b)} \min\{f_1(v),f_3(v)\}$.

Finally, consider $f_2(b)$, i.e., the minimum weight of a
claw-deletion set of $G_b^-$. For each $v$ in $C(b)$, there are
three possibilities: either $v$ belongs to the set, or $v$ does
not belong to the set and there is another vertex of $b \setminus
p_T(b)$ that does not belong to the set, or $v$ does not belong to
the set but any other vertex of $b \setminus p_T(b)$ belongs to
the set. We have to consider the third possibility for every $v$
in $C(b)$, adding $f_2(v)$ to the sum of $f_1(v')$ for every other
$v'$ in $C(b)$, and in that case adding also $s(b)$. For the first
two possibilities, the situation is similar to the one in the
computation of $f_3(b)$. This justifies the formula $f_2(b)=\min
\{ \sum_{v\in C(v)} \min\{f_1(v),f_3(v)\}, \min_{v_1 \in C(v)}
(s(b)+f_2(v_1) + \sum_{v\in C(v)\setminus\{v_1\}} f_1(v))\}$.
\qed\end{proof}

In Theorem~\ref{theo:timeBlockGraph} we analyze the time to
compute $f(r)$.

\begin{theorem}\label{theo:timeBlockGraph}
Let $G$ be a weighted connected block graph with $n$ vertices.
Given a block-cutpoint-tree of $G$, the minimum weight of a
claw-deletion set of $G$ can be determined in $\mathcal{O}(n)$
time.
\end{theorem}

\begin{proof}
If the graph $G$ is complete, the weight is zero. Otherwise, we
root the given block-cutpoint-tree $T$ of $G$ at a cutpoint $r$ of
$G$. Notice that $|V(T)|$ and $|E(T)|$ are $\mathcal{O}(n)$.

Then we compute bottom-up the functions $f_1$, $f_2$, $f_3$. The
computation for a leaf $b$ (recall that leaves of $T$ are blocks
of $G$) takes $\mathcal{O}(|b|-1)$ time. Notice that $|C(b)|$ is
also $\mathcal{O}(|b|-1)$ for a block $b$ of $G$ which is not a
leaf. Thus, the computation of $f_1(b)$ and $f_3(b)$ is also
$\mathcal{O}(|b|-1)$. We can compute (as a fourth function) the
difference $f_2(v)-f_1(v)$ for every cutpoint $v$ of $G$. So, for
the computation of $\min_{v_1 \in C(v)} (s(b)+f_2(v_1) +
\sum_{v\in C(v)\setminus\{v_1\}} f_1(v))$ we simply choose as
$v_1$ the vertex $v$ minimizing $f_2(v)-f_1(v)$. Therefore the
computation of $f_2(b)$ can be also done in $\mathcal{O}(|b|-1)$
time.

For the vertices $v$ which are cutpoints of $G$, we can compute
(as a fourth function) the difference $f_3(b)-f_1(b)$ for every
block $b$ of $G$. In this way, we can compute each of $f_1(v)$,
$f_2(v)$, and $f_3(v)$ in $\mathcal{O}(|C(v)|)$ time.

The whole complexity of the algorithm is then
$\mathcal{O}(|V(T)|+|V(G)|) = \mathcal{O}(n)$. \qed\end{proof}

Recall that a block-cutpoint-tree of a connected graph $G$ with
$n$ vertices and $m$ edges can be computed in $\mathcal{O}(n+m)$
time, as well as the connected components of a graph. This implies
\\

\noindent \textbf{Corollary~\ref{cor:timeBlockGraph}.} \emph{Let
$G$ be a weighted block graph with $n$ vertices and $m$ edges. The
minimum weight of a claw-deletion set of $G$ can be determined in
$\mathcal{O}(n+m)$ time.}

\paragraph*{\bf Proof of Running Time of Theorem~\ref{theo:cfvdTreewidth}.}
\emph{\textsc{Weighted Claw-free Vertex Deletion} can be solved in linear time on graphs with bounded treewidth. More precisely, there is a $2^{\mathcal{O}(k^2)} \cdot n$-time algorithm to solve \textsc{Weighted Claw-free Vertex Deletion} on $n$-vertex graphs $G$ with treewidth at most $k$.
}

\begin{proof}
We analyze the time to compute $\cdn_w[r,
\emptyset,\emptyset,\emptyset,\emptyset,\emptyset]$.
%
Since $tw(G) \leq k$, then $|X_t| = \mathcal{O}(k)$, for every node $t \in V(T)$.
For every leaf node $t$, function runs in constant time.

Let $t$ be an introduce node. Functions of Formulas (2.1) and
(2.3) run in constant time. Function (2.2) requires
$\mathcal{O}(k^{2.3728639})$ time~\cite{itai1978finding} for
checking if $N_{X_t \setminus S}(v) \text{ does not induce a }
\overline{K}_3$, $\mathcal{O}(|X_t \times X_t|) =
\mathcal{O}(k^2)$ for checking $\text{for every } (x,y) \in Z,$ if
$vx\in E(G) \mbox{ or } vy \notin E(G)$, $\mathcal{O}(|C|) =
\mathcal{O}(k)$ for checking if $N_{X_t}(v) \cap C =
\emptyset$, and $\mathcal{O}(|X_t \times X_t|) = \mathcal{O}(k^2)$
for the final condition. Such steps are executed for every
partition $S \cup A \cup B \cup C$ of $X_t$, which has $4^{\mathcal{O}(k)}$ possibilities, and for every $Z \subseteq X_t \times X_t$, which leads to $2^{\mathcal{O}(k^2)}$ choices of $Z$. Since the first is dominated by the latter, we obtain that computing $\cdn_w$ for an
introduce node requires $2^{\mathcal{O}(k^2)}$ time.

Let $t$ be a forget node. Formula (3.1) runs in $\mathcal{O}(1)$.
The minimum value asked for Formula (3.2) is obtained by checking every $(S,A',B',C',Z')$, which is bounded by the size of the power set of $X_t \times X_t$, $2^{\mathcal{O}(k^2)}$. Since all steps are executed for every partition $S \cup A \cup B \cup C$ of $X_t$ and for every $Z \subseteq X_t \times X_t$, the time required for a forget node $t$ is $2^{\mathcal{O}(k^2)}$.

Finally, let $t$ be a join node. Let $S \cup A \cup B \cup C$ be a
partition of $X_t$, and $Z \subseteq X_t \times X_t$. The value
asked for Formula (4) is obtained by the minimum sum of $\cdn_w$
in $t'$ and in $t''$, among all possibilities of $(A',B',C',Z')$
and $(A'',B'',C'',Z'')$, where the pair must satisfy (4). This
leads to a running time of $2^{\mathcal{O}(k^2)} \cdot
2^{\mathcal{O}(k^2)}\cdot \mathcal{O}(k)$. Those steps are
executed for every partition $S \cup A \cup B \cup C$ of $X_t$ and
for every $Z \subseteq X_t \times X_t$. Hence, the total running
time for computing $\cdn_w$ for a join node $t$ is bounded by
$2^{\mathcal{O}(k^2)}$.

Since the time to compute each entry of the table is upper bounded
by $2^{\mathcal{O}(k^2)}$ and the table has size
$2^{\mathcal{O}(k^2)} \cdot n$, the algorithm can be performed in
$2^{\mathcal{O}(k^2)} \cdot n$ time. This implies linear-time
solvability for graphs with bounded treewidth. \qed\end{proof}

\end{document}